\documentclass[11pt]{article}
\usepackage[utf8]{inputenc}
\usepackage{amsmath,amsthm,amssymb}
\usepackage{fullpage}
\usepackage{hyperref}
\usepackage{algorithm}
\usepackage[noend]{algorithmic}

\usepackage{tcolorbox}

\newcommand{\ba}{\mathbf{a}}
\newcommand{\bb}{\mathbf{b}}
\newcommand{\bc}{\mathbf{c}}
\newcommand{\be}{\mathbf{e}}
\newcommand{\bl}{\boldsymbol{\ell}}
\newcommand{\btheta}{\boldsymbol{\theta}}
\newcommand{\bmf}{\mathbf{f}}
\newcommand{\bo}{\mathbf{o}}
\newcommand{\bs}{\mathbf{s}}
\newcommand{\bx}{\mathbf{x}}
\newcommand{\by}{\mathbf{y}}
\newcommand{\bz}{\mathbf{z}}
\newcommand{\bp}{\mathbf{p}}
\newcommand{\bq}{\mathbf{q}}
\newcommand{\bfzero}{\mathbf{0}}

\newcommand{\caA}{\mathcal{A}}
\newcommand{\caB}{\mathcal{B}}
\newcommand{\caC}{\mathcal{C}}
\newcommand{\caO}{\mathcal{O}}
\newcommand{\E}{\mathbf{E}}
\newcommand{\bbR}{\mathbb{R}}
\DeclareMathOperator{\dist}{dist}

\DeclareMathOperator{\proj}{proj}
\DeclareMathOperator{\supp}{supp}
\DeclareMathOperator{\regret}{regret}
\DeclareMathOperator*{\argmin}{argmin}
\DeclareMathOperator*{\argmax}{argmax}

\usepackage{mathtools}
\DeclarePairedDelimiter{\abs}{\lvert}{\rvert}
\DeclarePairedDelimiter{\norm}{\lVert}{\rVert}

\newtheorem{theorem}{Theorem}[section]
\newtheorem{corollary}[theorem]{Corollary}
\newtheorem{lemma}[theorem]{Lemma}
\theoremstyle{definition}
\newtheorem{definition}[theorem]{Definition}
\newtheorem{remark}[theorem]{Remark}

\title{\bf No-regret algorithms for online $k$-submodular maximization}
\author{Tasuku Soma\\
The University of Tokyo\\
\texttt{tasuku\_soma@mist.i.u-tokyo.ac.jp}}
\date{\today}

\begin{document}
\maketitle

\begin{abstract}
We present a polynomial time algorithm for online maximization of $k$-submodular maximization.
For online (nonmonotone) $k$-submodular maximization, our algorithm achieves a tight approximate factor in an approximate regret.
For online monotone $k$-submodular maximization, our approximate-regret matches to the best-known approximation ratio, which is tight asymptotically as $k$ tends to infinity.
Our approach is based on the Blackwell approachability theorem and online linear optimization.
\end{abstract}

\section{Introduction}
\emph{Submodular functions} have a wide veriety of applications in combinatorial optimization, economics, communication, and machine learning~\cite{Fujishige2005,Krause2014survey}.
A set function $f: 2^V \to \bbR$ on a ground set $V$ is called a \emph{submodular function} if it satisfies $f(X) + f(Y) \geq f(X \cup Y) + f(X \cap Y)$ for all $X \subseteq V$. 
Equivalently, $f$ is submodular if it satisfies the \emph{diminishing return property}: $f(X \cup \{j\}) - f(X) \geq f(Y \cup \{j\}) - f(Y)$ for all $X \subseteq Y$ and $j \in V \setminus Y$.
In the last two decades, \emph{submodular maximization} has been studied extensively in theoretical computer science~\cite{Calinescu2011,Buchbinder2015}, machine learning~\cite{Krause2014survey}, and viral marketing~\cite{Kempe2003}. 
Although submodular maximization is NP-hard in general, constant-factor approximation algorithms have been devised for various constraints~\cite{Calinescu2011,Buchbinder2015}.

Recently, the paradigm of \emph{``optimization as a process''} has been proposed in the context of \emph{online learning}~\cite{Hazan2016book,Cesa2006}.
The goal of online learning is making a better decision in the face of uncertainty.
Formally, let us consider the following repeated two-player game between a player and an adversary.
At each $t$th round ($t \in [T] := \{1,\dots, T\})$, the player must select an action $x_t \in K$ (possibly in a randomized manner). 
After the choice of $x_t$, the adversary reveals a reward function $f_t: K \to [0,1]$ in the round, and the player gains $f_t(x_t)$.
The performance metric of the player's algorithm is the \emph{regret}: 
\begin{align}
\regret(f_1, \dots, f_T) = \max_{x \in K}\sum_{t \in [T]} f_t(x) - \sum_{t \in [T]} f_t(x_t).
\end{align}
That is, the regret is the difference between the player's total gain and the gain of the best fixed action in hindsight.
A player's algorithm is said to be \emph{no regret} if the expectation of the regret is sublinear: $\E[\regret(f_1,\dots,f_T)] = o(T)$, where the expectation is taken under the randomness in the player.

\emph{Online submodular maximization} is an online learning problem in which the action set is a set family $\caC \subseteq 2^V$ and the reward functions $f_t$ are submodular functions on $V$.
Since submodular maximization is NP-hard even in the offline setting, it is reasonable to relax the definition of the regret to the \emph{$\alpha$-regret}: 
\begin{align}
\regret_\alpha(f_1, \dots, f_T) = \alpha\max_{X \in \caC}\sum_{t \in [T]} f_t(X) - \sum_{t \in [T]} f_t(X_t),
\end{align}
where $\alpha > 0$ is a constant.
Intuitively, $\alpha$ corresponds to the offline approximation ratio.
A player's algorithm is said to be \emph{no $\alpha$-regret} if $\E[\regret_\alpha(f_1, \dots, f_T)] = o(T)$.
Streeter and Golovin~\cite{Streeter2009} presented the first no $(1-1/e)$-regret algorithm for online monotone submodular maximization under a cardinality constraint ($\caC$ is the set of subsets satisfying the cardinality constraint and $f_t$ are monotone submodular functions).
Golovin, Streeter, and Krause~\cite{Golovin2014} extended this algorithm to a matroid constraint, generalizing a well-known \emph{continuous greedy algorithm}~\cite{Calinescu2011}.
Recently, Roughgarden and Wang~\cite{Roughgarden2018} proposed no $1/2$-regret algorithm for (unconstrained) online nonmonotone submodular maximization.
Their algorithm is based on the \emph{double greedy algorithm}~\cite{Buchbinder2015};
at its core, they designed an online learning algorithm with two actions with a stronger regret guarantee.

\subsection{Our contribution}

\begin{figure}[t]
    \centering
    \begin{tcolorbox}[boxrule=0pt,toprule=.5pt,bottomrule=.5pt,sharp corners]
    \begin{algorithmic}
        \FOR{$t = 1, \dots, T$}
        \STATE A player (randomly) plays $\bx_t \in (k+1)^V$.
        \STATE An adversary reveals a $k$-submodular function $f_t: (k+1)^V \to [0,1]$ to the player as a value oracle.
        \STATE The player gains reward $f_t(\bx_t)$.
        \ENDFOR
    \end{algorithmic}
    \end{tcolorbox}
    \caption{The online $k$-submodular maximization protocol.}
\end{figure}

This paper examines online maximization of \emph{$k$-submodular} functions.
$k$-submodular functions are generalizations of submodularity and bisubmodularity, introduced by Huber and Kolmogolov~\cite{Huber2012}.
Formally, $k$-submodular functions are defined on $(k+1)^V = \{0,1,\dots, k\}^V$. 
A function $f : (k+1)^V \to \bbR$ is $k$-submodular if for any $\bx, \by \in (k+1)^V$, $f(\bx) + f(\by) \geq f(\bx \sqcup \by) + f(\bx \sqcap \by)$, where $\sqcup$ and $\sqcap$ are generalized ``union'' and ``intersection'' in $(k+1)^V$, respectively (see Section~\ref{sec:pre} for the formal definition).
Indeed, if $k = 1, 2$, $k$-submodularity is equivalent to submodularity and bisubmodularity, respectively.
The concepts of bisubmodularity and $k$-submodularity have numerous applications in valued CSP, delta matroids, generalized influence maximization, and image segmentation~\cite{Huber2012,Fujishige2005,Fujishige2005b,Ohsaka2015,Hirai2017}. 

For offline $k$-submodular maximization, Iwata, Tanigawa, and Yoshida~\cite{Iwata2016} gave a $1/2$-approximation algorithm.
The approximation ratio is tight even for $k=1$, i.e., submodular maximization~\cite{Feige2011}.
They also devised a $\frac{k}{2k-1}$-approximation algorithm for \emph{monotone} $k$-submodular maximzation and the approximation ratio is asymptotically tight. 

The main results of this paper are as follows:

\begin{itemize}
    \item For online $k$-submodular maximization, we devise a polynomial-time algorithm whose expected $1/2$-regret is bounded by $O(nk\sqrt{T})$, where $n = \abs{V}$.
This result generalizes the previous algorithm of Roughgarden and Wang~\cite{Roughgarden2018} for online submodular maximization.
    \item For online monotone $k$-submodular maximization, we present a polynomial-time algorithm whose expected $\frac{k}{2k-1}$-regret is $O(nk\sqrt{T})$.
\end{itemize}

To extend the algorithm of \cite{Iwata2016} to the online setting, we must consider an auxiliary online learning problem, which we call a \emph{$k$-submodular selection game}.
We show that it is sufficient to design an online algorithm for $k$-submodular selection games with a stronger regret guarantee, which is not obtained by using a standard online learning algorithm such as multiplicative weight update~\cite{Arora2012a}.
To this end, we exploit \emph{Blackwell's approachability theorem}\footnote{The possibility of using of Blackwell's approachability theorem was mentioned in Roughgarden and Wang~\cite{Roughgarden2018} without detail in a footnote. They designed an alternative algorithm for a similar problem without using Blackwell's theorem.}~\cite{Blackwell1956} and \emph{online linear optimization (OLO)}.
The Blackwell approachability theorem is a powerful generalization of von Neumann's minimax theorem for finite two-player games. 
In the online learning literature, the Blackwell approachability theory has been exploited to demonstrate the existence of no-regret algorithms for various problems, such as online learning with the internal and generalized regret, and well-calibrated forecasters (see \cite{Cesa2006} and references therein).
We exploit the Blackwell approachability theorem to design an algorithm with the desired stronger regret guarantee.
To obtain a concrete regret bound, we use a beautiful duality result between approachability and OLO~\cite{Abernethy2011}.
More precisely, we use their framework to obtain an online algorithm for $k$-submodular selection games by converting an OLO algorithm.

To demonstrate the flexibility of our approach based on Blackwell's theorem, we show that the algorithm for the nonmonotone case can be easily modified for the monotone case with a stronger approximation ratio $\frac{2k}{2k-1}$.
Furthermore, our algorithm and analysis work even for an \emph{adaptive adversary}.
An \emph{oblivious adversary} fixes $f_t$ ($t \in [T]$) before the first round, whereas an adaptive adversary can select $f_t$ after seeing $\bx_t$.
Since our approach is conceptually simpler than previous work~\cite{Roughgarden2018}, it almost immediately extends to an adaptive adversary.

\subsection{Related work}
An important special case of $k$-submodular functions is the \emph{bisubmodular} function.
Singh, Guillory, and Bilmes~\cite{Singh12} studied maximizing a bisubmodular function\footnote{Note that they used different terminology, \emph{directed bisubmodular} functions, to describe such functions.}. 
General $k$-submodular maximization was first studied by Buchbinder and \v{Z}ivn\'{y}~\cite{Ward2016}. 
They devised a $1/(1+\sqrt{k/2})$-approximation algorithm for $k$-submodular maximization.
Iwata, Tanigawa, and Yoshida~\cite{Iwata2016} presented a randomized algorithm with an improved and tight approximation factor of $1/2$ for $k$-submodular maximization.
A derandomized version of their algorithm was developed by Oshima~\cite{Oshima2018}.
Ohsaka and Yoshida~\cite{Ohsaka2015} studied monotone $k$-submodular maximization under a cardinality constraint.
Later, Sakaue~\cite{Sakaue2017} generalized it to a matroid constraint.

% Online learning and online convex optimization are vast fields among statistics, machine learning, optimization, and algorithmic game theory. 
% See books~\cite{Cesa2006,Hazan2016book} and references therein.
Online learning of discrete structure is called \emph{online structured learning}.
Efficient online algorithms were developed for various discrete structures, such as shortest paths and matroid basis~\cite{Takimoto2003,Suehiro2012}.
Most of these studies focused on optimizing \emph{linear} reward/loss functions (under a constraint), whereas our paper studies \emph{nonlinear} functions (without constraint).

\subsection{Organization}
The reminder of this paper is organized as follows.
Section~\ref{sec:pre} introduces $k$-submodularity, Blackwell's approachability theorem, and OLO.
Section~\ref{sec:nonmonotone} describes our algorithm for online $k$-submodular maximization along with $k$-submodular selection games.
Section~\ref{sec:monotone} presents our algorithm for online monotone $k$-submodular maximization.

%%%%%%%%%%%%%%%%%%%%%%%%%%%%%%%%%%%%%%%%%%%%%%%%%%
\section{Preliminaries}\label{sec:pre}
%%%%%%%%%%%%%%%%%%%%%%%%%%%%%%%%%%%%%%%%%%%%%%%%%%
\subsection{Notation}
For a positive integer $n$, we denote the set $\{1, \dots, n\}$ by $[n]$.
The probability simplex in $\bbR^k$ is denoted by $\Delta_k$.
The sets of nonnegative and nonpositive reals are denoted by $\bbR_+$ and $\bbR_-$, respectively.
The Euclidian norm is denoted by $\norm{\cdot}$.
The $j$th standard unit vector is denoted by $\be_j$.
The distance between a point $\bx$ and a set $S$ is defined as $\dist(\bx, S) := \inf_{\by \in S} \norm{\bx - \by}$.
The orthogonal projection of a point $\bx$ onto a set $S$ is denoted by $\proj_S(\bx)$.

\subsection{$k$-submodular functions}
Let $k$ be a positive integer. 
Throughout the paper, let $V = [n]$ be a ground set.
Define $(k+1)^V = \{0,1,\dots, k\}^V$.
For $\bx \in (k+1)^V$, we denote $\supp(\bx) = \{ j \in V : x(j) \neq 0 \}$.
% For a function $f : (k+1)^V \to \bbR$, $(X_1, \dots, X_k) \in (k+1)^V$ and $j \notin \cup_{i}X_i$, we define 
% \begin{align}
    % \Delta_{j,i}f(X_1, \dots, X_k) := f(X_1, \dots, X_{i-1}, X_i \cup \{j\}, X_{i+1}, \dots, X_k) - f(X_1, \dots, X_k).
% \end{align}
For a function $f : (k+1)^V \to \bbR$, $\bx \in (k+1)^V$, and $j \notin \supp(\bx)$, we define 
\begin{align}
    \Delta_{j,i}f(X_1, \dots, X_k) := f(\bx + i\be_j) - f(\bx),
\end{align}
where $\bx + i\be_j$ is a vector obtained by setting the $j$th entry of $\bx$ to $i$.
Since $x(j) = 0$, this is the standard addition in $\bbR^V$.
% It is often useful to regard an element in $(k+1)^V$ as a vector $\bx \in \{0, 1, \dots, k\}^n$ as $x(j) = i$ if and only if $j \in X_i$ ($x(j)=0$ conventionally denotes that $j$ is in none of $X_i$).
Let us define a binary operator $\sqcup$ and $\sqcap$ on $\{0, 1, \dots, k\}$ as 
\begin{align}
    i \sqcup i' &= 
    \begin{cases}
        \max\{i, i'\} & \text{if either $i = 0$, $i' = 0$ or $i = i'$} \\
        0            & \text{otherwise}
    \end{cases} \\
    i \sqcap i' &= 
    \begin{cases}
        \min\{i, i'\} & \text{if either $i = 0$, $i' = 0$ or $i = i'$} \\
        0            & \text{otherwise}
    \end{cases} 
\end{align}
We extend these binary operations to $(k+1)^V$ so that the operations are applied entry-wise:
for $\bx, \by \in (k+1)^V$, define $\bx \sqcup \by, \bx \sqcap \by \in (k+1)^V$ as
\begin{align}
    (\bx \sqcup \by)(j) &= x(j) \sqcup y(j) \quad (j \in V) \\
    (\bx \sqcap \by)(j) &= x(j) \sqcap y(j) \quad (j \in V) .
\end{align}
A function $f : (k+1)^V \to \bbR$ is \emph{$k$-submodular} if 
\begin{align}\label{eq:k-submod}
f(\bx) + f(\by) \geq f(\bx \sqcup \by) + f(\bx \sqcap \by)
\end{align}
for arbitrary $\bx, \by \in (k+1)^V$.
Ward and \v{Z}ivn\'{y}~\cite{Ward2016} showed that $k$-submodularity is equivalent to the following two conditions:
\begin{description}
    \item[Pairwise monotonicity] $\Delta_{j,i}f(\bx) + \Delta_{j,i'}f(\bx) \geq 0$ for $i \neq i'$, $\bx \in (k+1)^V$, and $j \notin \supp(\bx)$.
    \item[Orthant submodularity] $\Delta_{j,i}f(\bx) \geq \Delta_{j,i}f(\by)$ for $i$, $\bx\leq \by$, and $j \notin \supp(\by)$.
\end{description}
Define a partial order on $(k+1)^V$ by $\bx \leq \by$ if $\bx \sqcap \by = \bx$.
We say that $f : (k+1)^V \to \bbR$ is \emph{monotone} if $f(\bx) \leq f(\by)$ for arbitrary $\bx \leq \by$.

A vector $\bx \in (k+1)^V$ can be regarded as a $k$-subpartition of $V$.
That is, $(k+1)^V$ can be regarded as the set of $(X_1, \dots, X_k)$ ($X_i \subseteq V$, $X_i \cap X_{i'} = \emptyset$ if $i \neq i'$). 
The correspondence is given by $x(j) = i$ if and only if $j \in X_i$ (we conventionally regard that $x(j)=0$ if and only if $j$ is in none of $X_i$).
% In this set notation, $k$-submodularity~\eqref{eq:k-submod} can be written as 
% \begin{align}
    % f(X_1, \dots, X_k) + f(Y_1,\dots, X_k) \\
    % & \geq f((X_1\cup Y_1) \setminus (X_{-1} \cap Y_{-1}), \dots, (X_k\cup Y_k) \setminus (X_{-k} \cap Y_{-k})) + f(X_1 \cap Y_1, \dots, X_k \cap Y_k),
% \end{align}
% where $X_{-j} = \bigcup_{j' \neq j} X_{j'}$ ($i \in [k]$) and $Y_{-j}$ is defined similarly.
For $k = 1$, $k$-submodularity \eqref{eq:k-submod} is equivalent to submodularity, $f(X) + f(Y) \geq f(X \cup Y) + f(X \cap Y)$ for $X,Y \in 2^V$.
For $k = 2$, it is equivalent to bisubmodularity~\cite{Fujishige2005},
\begin{align}
    f(X_1, X_2) + f(Y_1, Y_2) \geq f((X_1\cup Y_1) \setminus (X_1 \cap Y_1), (X_2\cup Y_2)\setminus (X_2\cap Y_2)) + f(X_1 \cap Y_1, X_2 \cap Y_2),
\end{align}
for $(X_1, X_2), (Y_1, Y_2) \in 3^V$.
In \cite{Ward2016}, they showed that a submodular function $g: 2^V \to \bbR_+$ can be embedded into a bisubmodular function $f:3^V \to \bbR_+$ as
\begin{align}\label{eq:embed}
    f(S,T) = g(S) + g(V\setminus T) - g(T)
\end{align}
preserving the approximation ratio.
That is, if an $\alpha$-approximate maximizer of $f$ corresponds to an $\alpha$-approximate maximizer of $g$, for arbitrary $\alpha > 0$. 
This embedding demonstrates that our algorithm for online $k$-submodular maximization corresponds the algorithm of \cite{Roughgarden2018} for online submodular maximization.

A useful fact of $k$-submodular maximization is that there always exists a maximizer corresponding to a partition of $V$.
\begin{lemma}[\cite{Ward2016}]\label{lem:k-submod-partition}
    Let $k \geq 2$.
    For any $k$-submodular function $f$, there exists $\bo \in \argmax_{\bx \in (k+1)^V} f(\bx)$ such that $\supp(\bo) = V$.
\end{lemma}

\subsection{Blackwell's approachability theorem}
The celebrated Blackwell approachability theorem~\cite{Blackwell1956} is a powerful generalization of the von Neumann minimax theorem for two-player zero-sum games. 
Our presentation mostly follows \cite{Abernethy2011}.
%Let $I$ and $J$ be finite action sets, and $\bl(i,j) : I \times J \to \bbR^k$ be a vector valued reward function.
%We denote by $X$ and $Y$ the sets of mixed strategies in $I$ and $J$, respectively. We naturally extend $\bl$ to a biaffine function on $X \times Y$ as
Let $X \subseteq \bbR^m$ and $Y \subseteq \bbR^n$ be convex sets.
Let $\bl : X \times Y \to \bbR^k$ be a biaffine function, i.e, $\bl(\cdot, \by)$ is affine for any $\by \in Y$ and vice versa.
%\begin{align}
%    \bl(\bx, \by) := \sum_{i \in I}\sum_{j \in J}x(i)y(j)\bl(i,j).
%\end{align}
Let $S \subseteq \bbR^k$ be a closed convex set.
We call a tuple $(X,Y,\bl,S)$ a \emph{Blackwell instance}.
We say that:
\begin{itemize}
\item $S$ is \emph{satisfiable} if $\exists \bx \in X \forall \by \in Y: \bl(\bx,\by) \in S$.
\item $S$ is \emph{response-satisfiable} if $\forall \by \in Y \exists \bx \in X: \bl(\bx,\by) \in S$.
\item $S$ is \emph{halfspace-satisfiable} if an arbitrary hyperplane $H$ containing $S $ is satisfiable. 
\item $S$ is \emph{approachable} if there exists a sequence $(\bx_t)_{t \in [T]} \subseteq X$ such that for any sequence $(\by_t)_{t \in [T]} \subseteq Y$, $\dist\left(\frac{1}{T}\sum_{t \in [T]} \ell(\bx_t, \by_t), S \right) \to 0$ as $T \to \infty$.
\end{itemize}

\begin{theorem}[{The Blackwell approachability theorem~\cite{Blackwell1956}}]\label{thm:Blackwell}
    For a Blackwell instance $(I,J,\bl,S)$, the following conditions are equivalent:
    \begin{enumerate}
        \item $S$ is approachable.
        \item $S$ is halfspace-satisfiable.
        \item $S$ is response-satisfiable.
    \end{enumerate}
\end{theorem}

A \emph{halfspace oracle} $\caO$ is an oracle that takes a halfspace $H$ with $S \subseteq H$ as input and returns $\caO(H) = \bx_H \in X$.
A halfspace oracle is said to be \emph{valid} if $\bl(\bx_H, \by) \in H$ for any $\by \in Y$.
Note that the existence of a valid halfspace oracle is equivalent to the halfspace-satisfiability of $S$.
Even if a valid halfspace oracle exists, its efficient computation depends on the geometry of the feasible regions $X$ and $Y$.
If $X$ and $Y$ are polytopes, then a halfspace oracle can be constructed by linear programming (LP) as follows.

Let $H := \{\bz : \btheta^\top \bz \geq \beta \}$ be a halfspace.
Since $\bl$ is biaffine, $\btheta^\top \bl(\bx,\by) = \bx^\top P \by + \bb^\top\by + c$ for some matrix $P$, a vector $\bb$, and a constant $c$.
For computing a valid halfspace oracle, we can assume that $c = 0$ without loss of generality.
Then, $\bx_H$ is a response of a valid halfspace oracle if and only if $\bx_H \in \argmax_{\bx \in X} \min_{\by \in Y} (\bx^\top P \by + \bb^\top\by)$.
Let $Y = \{\by : A\by \geq \bc\}$. 
By the LP duality, the inner minimization $\min_{\by \in Y} (P^\top\bx+\bb)^\top \by$ is equivalent to the following dual:
\begin{align}
    \max \bc^\top \bq & \quad \text{s.t.} \quad A^\top \bq = P^\top\bx + \bb,\, \bq \geq \bfzero.
\end{align}
Since $X$ is also a polytope, after adding a constraint $\bx \in X$, we still have an LP.

\subsubsection{Online linear optimization and approachability}
The beauty of Blackwell's approachability theory is that it provides an algorithm for finding an approaching sequence, given a valid halfspace oracle.
Abernethy and Hazan~\cite{Abernethy2011} connected the approachability and OLO.
In OLO, we are given a fixed compact convex set $K \subseteq \bbR^k$.
In each $t$th round of OLO, a player selects $\bx_t \in K$.
Then an adversary reveals a vector $\bmf_t$ such that $\max_{\bx \in K} \abs{\bmf^\top \bx}\leq 1$.
The goal of the player is to minimize the regret:
\begin{align}
    \regret(\bmf_1, \dots, \bmf_T) = \sum_{t \in [T]} \bmf_t^\top \bx_t - \min_{\bx \in K} \sum_{t \in [T]} \bmf_t^\top \bx 
\end{align}
They devised an elegant algorithm for approachability, given a valid halfspace oracle $\caO$ and an algorithm $\caA$ for OLO, under the assumption that $S$ is a cone.

\begin{theorem}[Abernethy and Hazan~\cite{Abernethy2011}]\label{thm:Abernethy2011}
    Given a valid halfspace oracle $\caO$, a value oracle of $\bl$, a cone $S$, and an OLO algorithm $\caA$ on the polar cone $S^\circ$, there exists an algorithm $\caB$ that given a sequence $(\by_t)_{t\in[T]}$, computes a sequence $(\bx_t)_{t \in [T]}$ satisfying
    \begin{align}
         \dist\left( \frac{1}{T}\sum_{t\in[T]} \bl(\bx_t,\by_t), S \right) \leq \frac{1}{T}\regret_\caA(\bmf_1, \dots, \bmf_T), 
    \end{align}
         where $\bx_t = \caB(\by_1, \dots, \by_{t-1})$ and $\bmf_t = - \bl(\bx_t,\by_t)$ ($t \in [T]$).
\end{theorem}

We use \emph{online gradient descent}~\cite{Zinkevich2003} as a standard OLO algorithm. See Algorihm~\ref{alg:OGD} for the detail.
\begin{algorithm}[t]
    \caption{Online Gradient Descent for OLO~\cite{Zinkevich2003}}\label{alg:OGD}
    \begin{algorithmic}
        \REQUIRE{a compact convex set $K \subseteq \bbR^k$ and learning rate $\eta > 0$.}
        \STATE Let $\bx_0 \in K$ be an arbitrary point.
        \FOR{$t \in [T]$} 
        \STATE Play $\bx_t$ and observe $\bmf_t$.
        \STATE Let $\by_{t+1} = \bx_t - \eta\bmf_t$ and $\bx_{t+1} = \proj_K(\by_{t+1})$.
        \ENDFOR
    \end{algorithmic}
\end{algorithm}

\begin{theorem}[Zinkevich~\cite{Zinkevich2003}]\label{thm:OGD}
    Online gradient descent with learning rate $\eta > 0$ satisfies
    \begin{align}
        \regret(\bmf_1, \dots, \bmf_T) \leq \frac{1}{\eta} D^2 + \eta \sum_{t\in[T]} \norm{\bmf_t}^2,
    \end{align}
    where $D$ is the diameter of $K$.
\end{theorem}

%%%%%%%%%%%%%%%%%%%%%%%%%%%%%%%%%%%%%%%%%%%%%%%%%%
\section{No $1/2$-regret algorithm for $k$-submodular maximization}\label{sec:nonmonotone}
%%%%%%%%%%%%%%%%%%%%%%%%%%%%%%%%%%%%%%%%%%%%%%%%%%
In this section, we present our algorithm for online $k$-submodular maximization.

\subsection{$k$-submodular selection game}
Let us consider the following online learning problem, which we call a \emph{$k$-submodular selection game}.
In the $t$th round of the game, a player predicts a probability vector $\bp_t \in \Delta_k$.
An adversary's play is $\by_t = (\ba_t, \bb_t) \in Y$, where $Y$ is the set of $(\ba, \bb) \in [-1,1]^k \times [-1,1]^k$ such that 
\begin{align*}
    a(i) + a(i') &\geq 0 && (i \neq i')\\
    b(i) + b(i') &\geq 0 && (i \neq i')\\
    b(i) &\geq a(i)     && (i \in [k]).
\end{align*}
The feedback to the player is only $\bb_t$.
We denote the set of the adversary'play by $Y$. 
For a fixed $\bb$, we denote $Y(\bb) = \{\ba \in [-1,1]^k : (\ba, \bb) \in Y \}$.

\begin{definition}
    Let $\alpha > 0$.
    An online algorithm $\caA$ is an $\alpha$-selection algorithm for a $k$-submodular selection game with rate $g(k,T)$ if it satisfies 
    \begin{align}
    \max_{i^* \in [k]} \sum_{t \in [T]} a_t(i^*) -  \sum_{t\in [T]} \sum_{i\in[k]} (\alpha\cdot b_t(i) + a_t(i)) p_t(i)  \leq g(k,T),
    \end{align}
    where $g(k,T)$ is sublinear in $T$.
\end{definition}

Our main result is as follows.
\begin{theorem}\label{thm:k-submod-selection}
    There exists a $1$-selection algorithm for a $k$-submodular selection game with rate $g(k,T) = O(k\sqrt{T})$.
\end{theorem}
To prove this theorem, we appeal to the Blackwell approachability theorem. 
First, we define a biaffine vector reward function $\bl$:
For $\bp \in \Delta_k$ and $\by = (\ba, \bb) \in Y$, let
\begin{align}
    \ell(\bp, \by)(i) =  a(i) - \sum_{i' \in [k]} (b(i')+a(i'))p(i').
\end{align}
Then, $S = \bbR_-^k$ is approachable in a Blackwell instance $(\Delta_k, Y, \bl, S)$  if and only if a $1$-selection algorithm exists for a $k$-submodular selection game.
We now show that $S$ is approachable.
By the Blackwell approachability theorem, it suffices to show that $S$ is response-satisfiable.
Indeed, this fact is already observed in \cite{Iwata2016}.

\begin{lemma}[{\cite[Theorem~2.1]{Iwata2016}}]
For a fixed adversary's play $(\ba, \bb)$, there exists $\bp \in \Delta_k$ that only depends on $\by$ and satisfies
\begin{align}
    \max_{i^* \in [k]} a(i^*) - \sum_{i\in[k]} (b(i) + a(i)) p(i) \leq 0.
\end{align}
% More precisely, $\bp$ can be computed by Algorihm~\ref{alg:halfspace-oracle}.
\end{lemma}

% \begin{algorithm}[t]
    % \caption{A valid halfspace-oracle for $k$-submodular selection game~\cite{Iwata2016}}\label{alg:halfspace-oracle}
    % \begin{algorithmic}[1]
        % \REQUIRE $\bb \in [-1,1]^k$ satisfying \eqref{}
        % \STATE Sort $\bb$ so that $b(1) \geq b(2) \geq \dots \geq b(k)$.
        % \STATE Let $i^+$ be the maximum $i \in [k]$ such that $b(i) > 0$. If such a $i$ do not exist, let $i^+ = 0$.
        % \IF{$i^+ \leq 1$}
        % \STATE Let $p(1) = 1$ and $p(i) = 0$ ($i \neq 1$).
        % \ELSIF{$i^+ = 2$}
        % \STATE Let $p(i) = \frac{b(i)}{b(1)+b(2)}$ ($i = 1, 2$) and $p(i) = 0$ ($i \geq 3$).
        % \ELSE
        % \STATE Let $p(i) =
        % \begin{cases}
            % 2^{-i} & \text{if $i < i^+$} \\
            % 2^{-i^+ + 1} & \text{if $i = i^+$} \\
            % 0 & \text{otherwise}
        % \end{cases} 
        % $ for $i \in [k]$.
        % \ENDIF
        % \RETURN $\bp$.
    % \end{algorithmic}
% \end{algorithm}
% Note that Algorithm~\ref{alg:halfspace-oracle} provides a valid halfspace oracle. 

Therefore, the Blackwell approachability theorem implies the existence of a no-regret algorithm for a $k$-submodular selection game.
In particular, exploiting the result of \cite{Abernethy2011}, we obtain Algorithm~\ref{alg:k-submod-selection} for a $k$-submodular selection game.

\begin{algorithm}[t]
    \caption{A $1$-selection algorithm for a $k$-submodular selection game}\label{alg:k-submod-selection}
    \begin{algorithmic}[1]
        \REQUIRE An OLO algorithm $\caA$ with feasible region $K := \{ \btheta \in \bbR^k_+ : \norm{\btheta} \leq 1 \}$.
        \STATE Set up $\caA$.
        \FOR{$t \in [T]$}
        \STATE $\btheta_t \gets \caA(\bmf_1, \dots, \bmf_{t-1})$, where $\bmf_s := -\hat{\bl}_s$ ($s \in [t-1]$).
        \STATE Solve LP
        \begin{align}\label{eq:halfspace-LP}
            \bp_t \in \argmin_{\bp \in \Delta_k} \max_{\by \in Y} \btheta^\top\bl(\bp,\by)
        \end{align} 
        to obtain $\bp_t$. 
        \STATE Play $\bp_t$ and observe $\bb_t$.
        \STATE For $i \in [k]$, let $\hat{\bl}_t$ be a vector such that $\hat{\ell}_t(i) := \max_{\ba_t \in Y(\bb_t)} \ell(\bp_t, (\ba_t, \bb_t))(i)$.
        \ENDFOR 
    \end{algorithmic}
\end{algorithm}

\begin{lemma}\label{lem:k-submod-selection}
    Algorithm~\ref{alg:k-submod-selection} satisfies
    \begin{align}
        \max_{i^* \in [k]} \sum_{t \in [T]} a_t(i^*) - \sum_{t\in [T]} \sum_{i\in[k]} (b_t(i) + a_t(i)) p_t(i) 
        \leq \regret_\caA(\bmf_1, \dots, \bmf_T),
    \end{align}
    for any $(\ba_t, \bb_t) \in Y$ $(t \in [T])$,
    where $\regret_\caA (\bmf_1, \dots, \bmf_T) = \sum_{t\in[T]} \bmf_t^\top \btheta_t - \min_{\btheta \in K} \sum_{t \in [T]} \bmf_t^\top \btheta$ is the regret of the OLO algorithm $\caA$.
\end{lemma}
\begin{proof}
    The proof mostly follows from \cite{Abernethy2011}, but we provide the full proof for the sake of completeness.
    Since $S$ is halfspace-satisfiable, LP~\eqref{eq:halfspace-LP} has a solution.
    Indeed, solving LP~\eqref{eq:halfspace-LP} simply computes an output of a valid halfspace oracle for a halfspace $H_t = \{\bx \in \bbR^k : \btheta_t^\top \bx \leq 0 \}$.
    Let us fix arbitrary $\by_t = (\ba_t, \bb_t) \in Y$ ($t \in [T]$).
    Then,
    \begin{align*}
        \dist\left(\frac{1}{T}\sum_{t \in [T]} \bl(\bp_t, \by_t), S \right) 
        &= \max_{\btheta \in K} \frac{1}{T}\sum_{t \in [T]} \bl(\bp_t, \by_t) ^\top \btheta \\
        &\leq \max_{\btheta \in K} \left[ \frac{1}{T}\sum_{t \in [T]} \hat{\bl}_t^\top \btheta \right] \\
        &= \max_{\btheta \in K} \left[ -\frac{1}{T}\sum_{t \in [T]} \bmf_t^\top \btheta \right] \\
        &\leq \frac{1}{T}\max_{\btheta \in K} \left[\sum_{t \in [T]}\bmf_t^\top \btheta_t  -\sum_{t \in [T]} \bmf_t^\top \btheta \right] 
        \tag{Since $\bmf_t^\top \btheta_t = -\btheta_t^\top \hat{\bl}_t \geq 0$ by the valid halfspace oracle property} \\
        &= \frac{\regret_\caA(\bmf_1, \dots, \bmf_T) }{T}. 
    \end{align*}
    Now the claim of the lemma is immediate from the following:
    \begin{align*}
        \frac{1}{T}\left[ \max_{i^* \in [k]} \sum_{t \in [T]} a_t(i^*) - \sum_{t\in [T]} \sum_{i\in[k]} (b_t(i) + a_t(i)) p_t(i) \right] 
        \leq \dist\left(\frac{1}{T}\sum_{t \in [T]} \bl(\bp_t, \by_t), S \right) 
    \end{align*}
\end{proof}

\begin{proof}[Proof of Theorem~\ref{thm:k-submod-selection}]
    We can use online gradient descent as an internal OLO algorithm $\caA$, which satisfies 
    \begin{align} 
        \regret_\caA(\bmf_1, \dots, \bmf_T) 
        \leq \frac{1}{\eta} D^2 + \eta \sum_{t\in[T]} \norm{\bmf_t}^2
        \leq \frac{1}{\eta} O(k) + \eta O(kT)
    \end{align} 
    where we used that $D = O(\sqrt{k})$ is the diameter of $\Delta_k$ and $\norm{\bmf_t}^2 = O(k)$ for $t \in [T]$ in the second inequality.
    Setting $\eta = O(1/\sqrt{T})$, we obtain the regret bound $O(k\sqrt{T})$.
    Combined with Lemma~\ref{lem:k-submod-selection}, we see that Algorithm~\ref{alg:k-submod-selection} is a $1$-selection algorithm with rate $O(k\sqrt{T})$.
\end{proof}

\begin{remark}
    Since Algorithm~\ref{alg:k-submod-selection} is deterministic if we use online gradient descent as an internal OLO algorithm, the guarantee in Theorem~\ref{thm:k-submod-selection} holds even for an adaptive adversary.
\end{remark}

\subsection{Main algorithm}
Now we present our main algorithm for online $k$-submodular maximization.

\begin{algorithm}[h]
    \caption{No $1/(\alpha+1)$-regret algorithm for $k$-submodular maximziation}\label{alg:k-submod}
\begin{algorithmic}[1]
    \REQUIRE $\alpha$-selection algorithms $\caA_j$ for a $k$-submodular selection game ($j \in [n]$).
    \STATE Set up $\caA_j$ ($j \in [n]$).
    \FOR{$t = 1, \dots, T$}
    \STATE Set $\bx_t^{(0)} := \bfzero$.
    \FOR{$j \in [n]$}
    \STATE Receive $\bp_t^{(j)} \in \Delta_k$ from $\caA_j$.
    \STATE Sample $i \in [k]$ from the probability distribution $\bp_t^{(j)}$, and set $\bx_t^{(j)}:= \bx_t^{(j-1)}+ i\be_j$. 
    \ENDFOR
    \STATE Play $\bx_t = \bx_t^{(n)}$ and receive $f_t$.
    \FOR{$j \in [n]$}
    \STATE Feedback $b_{t}^{(j)}(i) := \Delta_{j,i}f_t(\bx^{(j-1)})$ ($i \in [k]$) to $\caA_j$.
    \ENDFOR
    \ENDFOR
\end{algorithmic}
\end{algorithm}

\begin{theorem}\label{thm:k-submod}
    Given $\alpha$-selection algorithms $\caA_j$ ($j \in [n]$) for $k$-submodular selection games with rate $g(k,T)$, Algorithm~\ref{alg:k-submod} achieves
    \begin{align}
        \E \left[ \frac{1}{\alpha+1}\max_{\bo \in (k+1)^V}\sum_{t\in[T]} f_t(\bo) - \sum_{t\in[T]} f_t(\bx_t) \right] \leq n g(k,T),
    \end{align} 
    where the expectation is taken under the randomness in Algorithm~\ref{alg:k-submod}.
\end{theorem}
\begin{proof}
    Let $\bo \in (k+1)^V$ be an optimal solution such that $\supp(\bo) = [n]$ (such an optimal solution exists by Lemma~\ref{lem:k-submod-partition}).
    For each $t \in [T]$ and $j = 0, 1, \dots, n$, let $\bo_t^{(j)} := (\bo \sqcup \bx_t^{(j)}) \sqcup \bx_t^{(j)}$.
    Note that $\bo_t^{(0)} = \bo$ and $\bo_t^{(n)} = \bx_t^{(n)}$.
    Let $\bs_t^{(j-1)}$ be a vector obtained by setting the $j$th element of $\bo_t^{(j-1)}$ to $0$ for $j \in [n]$.
    Define $a_t^{(j)}(i) := \Delta_{j,i} f_t(\bs_t^{(j-1)})$ and $b_t^{(j)}(i) := \Delta_{j,i} f_t(\bx_t^{(j-1)})$.
    By orthant submodularity and pairwise monotonicity, we have
\begin{align*}
    a_t^{(j)}(i) + a_t^{(j)}(i') &\geq 0 && (i \neq i')\\
    b_t^{(j)}(i) + b_t^{(j)}(i') &\geq 0 && (i \neq i')\\
    b_t^{(j)}(i) &\geq a_t^{(j)}(i)     && (i \in [k]).
\end{align*}
Therefore, $\bb_t^{(j)}$ is valid feedback to $\caA_j$ ($j \in [n]$).
    Let us fix $j \in [n]$ and let $i^* := o(j)$.
    Note that $i^* \in [k]$, since $\supp(\bo) = [n]$.
    Since $\caA_j$ is an $\alpha$-selection algorithm, we have 
    \begin{align}\label{eq:1-selection-guarantee}
    \sum_{t \in [T]}\sum_{i\in[k]} (a_t^{(j)}(i^*) - a_t^{(j)}(i))p_t^{(j)}(i)
    \leq \alpha \sum_{t\in [T]} \sum_{i\in[k]} b_t^{(j)}(i)p_t^{(j)}(i)  +  g(k,T),
    \end{align}
    conditioned on $\bx_t^{(j-1)}$ ($t \in [T]$).
    Taking the expectation on $\bx_t^{(j-1)}$ ($t \in [T]$), we obtain
    \begin{align}
    \E\left[ \sum_{t \in [T]} (f_t(\bo_t^{(j-1)}) - f_t(\bo_t^{(j)})) \right]
    \leq \alpha \E\left[ \sum_{t\in [T]} (f_t(\bx_t^{(j)}) - f_t(\bx_t^{(j-1)})) \right]  +  g(k,T).
    \end{align}
    Summing these inequalities for $j \in [n]$, we arrive at
    \begin{align*}
    \E\left[ \sum_{t \in [T]} (f_t(\bo) - f_t(\bx_t)) \right]
    &\leq \alpha\E\left[ \sum_{t\in [T]} (f_t(\bx_t) - f_t(\bfzero)) \right]  + ng(k,T)\\
    &\leq \alpha\E\left[ \sum_{t\in [T]} f_t(\bx_t) \right]  + n g(k,T) 
     \tag{since $f_t(\bfzero)\geq 0$ ($t \in [T]$)},
    \end{align*}
    which proves the theorem.
\end{proof}

Combining this theorem with Lemma~\ref{lem:k-submod-selection}, we obtain the main result.
\begin{corollary}
    There exists a polynomial-time algorithm for online $k$-submodular maximization whose $1/2$-regret is bounded by $O(kn\sqrt{T})$.
\end{corollary}

\begin{remark}
    Since Algorithm~\ref{alg:k-submod-selection} is deterministic, \eqref{eq:1-selection-guarantee} is valid for an adaptive adversary.
    Therefore, the regret bound of Algorithm~\ref{alg:k-submod} holds for an adaptive adversary.
    Note that a selection algorithm used in Roughgarden and Wang~\cite{Roughgarden2018} is randomized;
    therefore it requires different analysis for an adaptive adversary.
\end{remark}

%%%%%%%%%%%%%%%%%%%%%%%%%%%%%%%%%%%%%%%%%%%%%%%%%%
\section{Online monotone $k$-submodular maximization}\label{sec:monotone}
%%%%%%%%%%%%%%%%%%%%%%%%%%%%%%%%%%%%%%%%%%%%%%%%%%
To demonstrate the flexibility of our method with the Blackwell approachability theory, we present a no $\frac{k}{2k-1}$-regret algorithm for online monotone $k$-submodular maximization.
To this end, we define a modified version of a $k$-submodular selection game, which we call a \emph{monotone $k$-submodular selection game}.
The only difference in the monotone case is that the set of the adversary's play is further restricted to  $Y_+ := Y \cap (\bbR_+^k \times \bbR_+^k)$, which means that $\by_t \geq \bfzero$.

\begin{lemma}
    There exists a $(1-1/k)$-selection algorithm for a monotone $k$-submodular selection game with rate $g(k,T) = O(k\sqrt{T})$.
\end{lemma}
\begin{proof}
Again, we use the Blackwell approachability theorem.
We define a slightly modified vector reward function $\bl'$ as follows:
\begin{align}
    \ell'(\bp, \by)(i) =  a(i) - \sum_{i' \in [k]} (\alpha\cdot b(i')+a(i')) p(i'),
\end{align}
where $\alpha = 1-1/k$.
It suffices to show that $S = \bbR_-^k$ is response-satisfiable for a Blackwell instance $(X, Y_+, \bl', S)$.
In \cite[Theorem~2.2]{Iwata2016}, it is shown that for fixed $\by = (\ba, \bb) \in Y_+$, there exists $\bp \in \Delta_k$ such that $\bl'(\bp, \by) \leq \bfzero$.
Therefore, there exists an online algorithm for producing an approaching sequence.
Indeed, such an algorithm can be constructed by a slight modification of Algorithm~\ref{alg:k-submod-selection}:
instead of $\bl$ and $Y$, we use $\bl'$ and $Y_+$, respectively.
It is easy to see that the modified algorithm produces a sequence $\bp_t$ ($t\in[T]$) with the same guarantee as in Lemma~\ref{lem:k-submod-selection}:
\begin{align}
        \max_{i^* \in [k]} \sum_{t \in [T]} a_t(i^*) - \sum_{t\in [T]} \sum_{i\in[k]} (\alpha\cdot b_t(i) + a_t(i)) p_t(i) 
        \leq \regret_\caA(\bmf_1, \dots, \bmf_T),
\end{align}
    for any $(\ba_t, \bb_t) \in Y$ $(t \in [T])$, where $\caA$ is an internal OLO algorithm.
    Again, using online gradient descent as $\caA$, we obtain the same bound as before, which completes the proof.
\end{proof}

Combining this result with Theorem~\ref{thm:k-submod}, we obtain the following.
\begin{theorem}
    There exists a polynomial-time algorithm for online monotone $k$-submodular maximization whose $\frac{k}{2k-1}$-regret is bounded by $O(kn\sqrt{T})$.
\end{theorem}
\begin{proof}
    We use the same notation as in the proof of Theorem~\ref{thm:k-submod}.
    Since $f_t$ is monotone ($t \in [T]$), we have $\ba_t^{(j)}, \bb_t^{(j)} \geq \bfzero$ ($t\in[T]$, $j \in [n]$).
    Therefore, $\bb_t^{(j)}$ is valid feedback to an algorithm for a monotone $k$-submodular selection game.
    Since $\alpha = 1-1/k$, we have the same bound for the $\frac{k}{2k-1}$-regret.
\end{proof}

\subsection*{Acknowledgement}
The author thanks Takanori Maehara, Shinsaku Sakaue, Yuichi Yoshida, and Kaito Fujii for valuable discussions.
The author also thanks Tim Roughgarden and Joshua R. Wang for sharing a draft of \cite{Roughgarden2018}.
This work was supported by ACT-I, JST.

\bibliographystyle{IEEEtranS}
\bibliography{main.bib}
\end{document}